\documentclass[11pt]{article}

\usepackage[margin=1.1in]{geometry}
\usepackage{amsmath,amsthm}

\usepackage{xcolor}
\usepackage{graphicx}
\usepackage{lineno}
\usepackage{alltt}
\usepackage{mathtools,amssymb}
\usepackage{statmath}
\usepackage{natbib}
\usepackage{booktabs}
 \usepackage{setspace}
\usepackage{epsf}
\usepackage{graphicx}
\graphicspath{{./graphics/}{.}}
\usepackage{anyfontsize} % avoid font shape warnings
\usepackage[nolists]{endfloat}
\newtheorem{thm}{Theorem}[section]

\newcommand{\cip}{\mbox{$\perp\!\!\!\perp$}}

\newcommand{\nothere}[1]{}

\title{A Joint QoL–Survival Framework with Debiased Estimation under Truncation by Death}

\author{Torben Martinussen$^{1,*}$, 
Klaus K\"ahler Holst$^{2}$,
Christian Bressen Pipper$^{2,3}$, 
Per Kragh Andersen$^{1}$ \\
$^{1}$Section of Biostatistics, University of Copenhagen, Copenhagen, Denmark\\
$^{2}$Novo Nordisk, Søborg, Denmark \\
$^{3}$ Section of Epidemiology, Biostatistics and Biodemography, University of Southern Denmark, Odense, Denmark
}

\begin{document}

\maketitle

%\doublespacing 

\vspace{1cm}

\centerline{\sc Summary}
Evaluating quality-of-life (QoL) outcomes in populations with high mortality risk is complicated by truncation by death, since QoL is undefined for individuals who do not survive to the planned measurement time. We propose a framework that jointly models the distribution of QoL and survival without extrapolating QoL beyond death. Inspired by multistate formulations, we extend the joint characterization of binary health states and mortality to continuous QoL outcomes. Because treatment effects cannot be meaningfully summarized in a single one-dimensional estimand without strong assumptions, our approach simultaneously considers both survival and the joint distribution of  QoL and survival with the latter conveniently displayed in a simplex. We develop assumption-lean, semiparametric estimators based on efficient influence functions, yielding flexible, root-n consistent estimators that accommodate machine-learning methods while making transparent  the conditions these must satisfy. The proposed method is illustrated through simulation studies and two real-data applications.
\noindent

\vspace{-2mm}


\bigskip

\section{Introduction}

The evaluation of functional outcomes, such as quality of life (QoL), in studies involving patients with high severity of illness or risk of mortality presents significant challenges, as the outcome is only well-defined for subjects who survive until the planned landmark time for measurement. This phenomenon, known as \emph{truncation by death}, will drive the clinical questions of interest and is therefore a crucial component of both study design and subsequent data analysis.
Addressing this complex issue requires careful consideration, as different statistical approaches support  fundamentally different clinical research questions and therefore complement each other rather than necessarily agree in terms of conclusions. This premise is especially relevant  when treatment affects the terminal event process. When the purpose is to evaluate a treatment effect on an endpoint such as QoL, we argue, similarly to \citep{sjolander2011_principalstrat}, that the choice of method should be guided by whether actionable and clinical relevant decisions can be made based on the estimand.

Existing methodologies broadly fall into categories that either treat measurements truncated by terminal events as ordinary missing data that can be handled using specific missing at random assumptions, causal methods that aim at estimating the treatment effect in specific subpopulations such as always survivors, and finally methods that model the QoL outcome conditional on not having reached the terminal event.
Examples of the first approach include direct modeling of the joint trajectory of both the functional outcome and survival \citep{diggle2002, wulfsohn1997, rizopoulos2012}. Such approaches, however, effectively extrapolate the QoL endpoint beyond death, which means the interpretation relies heavily on untestable model assumptions which are not desirable in many situations. Moreover, they address the effect of treatment in a hypothetical scenario where it is assumed that death can be prevented. Such a scenario may not be clinically relevant \citep{Kahan2020}.
Principal stratification methods and in particular Survivor Average Causal Effects (SACE) \citep{robins1986:sace, ding2011identifiability, wang2017identification, YangDing2018Biometrics, LuoLiHe2023Biometrics, stensrud2022:ice_estimands} are examples of the second approach. While the principal stratification provides a sound mathematical framework, it has been pointed out that the resulting treatment effect estimates can be difficult to base clinical decisions on, as the patient subgroups might not be identifiable at the time of intervention.
More recently, separable effects have been proposed as an alternative \citep{stensrud2023conditional, stensrud2022:ice_estimands}, aiming to quantify causal effects of modified versions of the study treatment within an observable subset of the population. These effects, including separable direct and indirect effects, attempt to disentangle causal mechanisms by considering hypothetically modified treatments that operate through distinct causal pathways. While addressing some limitations of SACE, this approach still depends on the conceptualization of imagined interventions (i.e., modified treatments and their isolation conditions) which may be hard to justify in practice. From a clinical transparency perspective, relying on such hypothetical interventions to define a treatment effect can therefore be challenging for practical decision-making, especially in the primary evaluation of treatment efficacy.
The last approach, i.e. modeling
the QoL outcome conditional on being alive \citep{kurland2005directly, kurland2009longitudinal}, has been criticized for lack of causal interpretation \citep{rubin2006:principalstrata} as it compares two different populations, survivors in the treated and non-treated group respectively. This means that any treatment effect deduced by contrasting these two populations cannot be viewed in isolation as it may in part be driven by the effect of treatment on survival. 

In contrast, the approach proposed here provides a framework for jointly describing the distribution of QoL and survival alongside the marginal distribution of survival.
This is inspired by existing multistate models that effectively characterize the joint evolution of binary QoL indicators and death. We extend this conceptualization to continuous QoL measurements, recognizing that a disease score is not clinically meaningful beyond the terminal event.
The inherent complexity of truncation by death makes it challenging — if not impossible — to summarize treatment effects in a single one-dimensional estimand without unverifiable assumptions, such as those invoked when targeting separable effects. Consequently, we emphasize the importance of jointly considering both the probability of avoiding the terminal event and the joint probability distribution of QoL and being alive.
The proposed methodology is assumption-lean and does not require explicit assumptions about the trajectory of QoL or the timing of the terminal event.
We further develop debiased estimators based on
corresponding efficient influence functions (EIF) \citep{bickel1993efficient,kennedy2022semiparametric}. This
give rise to efficient and flexible root-n consistent estimators that are asymptotically model-free
by enabling data-adaptive methods (e.g. machine learning). We also derive the corresponding remainder term, which clarifies the requirements placed on any machine-learning methods used in practice. The method is investigated in a simulation study using realistic sample sizes and these confirm our theoretical results. We apply the proposed method
estimators first to some data concerning quality of life for prostate cancer patients and then also to data from the FLOW clinical kidney outcome trial \citep{Flow2024},
where focus is on the  estimated glomerular filtration rate (eGFR).

\section{Quantifying treatment effects on survival and QoL jointly}

\subsection{Preliminaries}
We use $T$ to denote the time until death and let $Y$ denote the marker, such as quality of life (QoL).
The marker  is planned to be measured at landmark time point $t$,
 but it is unmeasured if $T < t$ (truncation by death). The treatment indicator is denoted by $A$ and baseline covariates by $L$.
 We %let $X=(A,L^T)^T$ and 
use 
 $(T^a,Y^a)$ to denote potential outcomes if treatment was set to $a$.
 Estimands will be formulated for the "full data" situation (no censoring), but the proposed methods allow for the common observed data situation of right-censoring. 
Let $C$ denote the right-censoring time and set $T^*=T\wedge C$.
 Since we are interested in the marker at the landmark point in time $t$, censoring only affects the needed information if it happens before time $t$.  
 %and emphasize that 
 %because of the land
 %and denote 
 Thus, we denote  an observed 
 %the observed data by $\{T^*=T\wedge C,\Delta=I(T\leq C),Y I(t\leq T^*),A,L\}$, OR RATHER 
 data point by $O=\{T_t^*=t\wedge T^*,\Delta_t=I(T\wedge t\leq C),I(t= T_t^*),Y I(t= T_t^*),A,L\}$ and
  assume that $\{T,Y I(t< T)\}$ and $C$ are conditionally independent, given $A$ and $L$.
 Although our main results apply to a general real-valued marker 
$Y$, we begin with the situation where $Y$ is binary  to fix ideas. The binary case also yields a convenient summary, which is described in the next subsection.

\subsection{Binary marker}

If the QoL marker is binary, for example, $Y=1$ for `good QoL' and $Y=0$ for 'poor QoL’, then its development, jointly with death, can be described using a multi-state model, say $V(t)$. The relevant model is the %{\em 
illness-death model
%}, 
%see Figure \ref{fig:illnessdeath}, 
with states `0: alive $Y=0$', `1: alive $Y=1$, and `$D$ dead', that is, $V(t)=D$ if and only if  $T\leq t$. 
%with $T$ denoting the survival time, i.e., $V(t)=D$ iff $T\leq t$.
%\medskip
%\centerline{Figure \ref{fig:illnessdeath} about here}

If treatment was set to $a$, the  corresponding 
state occupation probabilities
 ${\cal Q}_j^{a}=P\{V^{a}(t)=j\}, j=0,1,D$ (suppressing their dependencies on $t$),  for which 
${\cal Q}_0^{a}+{\cal Q}_1^{a}+{\cal Q}_D^{a}=1$,
can be depicted in a {\em simplex}, see Figure \ref{fig:simplex}. 
The `good' state to occupy is, obviously, state 1.
Here, a beneficial treatment effect, as exemplified in the figure, is indicated by $A=1$ being closer to $({\cal Q}_0,{\cal Q}_1,{\cal Q}_D)=(0,1,0)$ (lower right corner) and further away from $({\cal Q}_0,{\cal Q}_1,{\cal Q}_D)=(0,0,1)$ (top corner) than $A=0$. 
Key quantities are thus
$P(Y^{a}=1,T^{a}>t)$ and $P(T^{a}\leq t)$ corresponding to ${\cal Q}_1^{a}$ and ${\cal Q}_D^{a}$.
\medskip

\begin{figure}
\centering
\caption{Simplex summarizing $\{Q_0(t), Q_1(t), Q_D(t)\}$ for a fixed $t$.
\label{fig:simplex}}
\setlength{\unitlength}{5mm}
\makebox[\linewidth]{%
\begin{picture}(22,22)(0,0) % width and height now aligned to the drawing
  \put(0,4){\line(1,0){20}}
  \put(-3,4){\large $Q_0=1$}
  \put(21,4){\large $Q_1=1$}
  \put(8.5,20){\large $Q_D=1$}
  \put(15,12){\large $Q_0=0$}
  \put(2,12){\large $Q_1=0$}
  \put(8.5,2){\large $Q_D=0$}
  \put(12,8){\circle*{0.3}}\put(12.5,8){\large $A=1$}
  \put(8,12){\circle*{0.3}}\put(8.5,12){\large $A=0$}
  \put(0,4){\line(2,3){10}}
  \put(20,4){\line(-2,3){10}}
  \put(20,4){\circle*{0.5}}
  \put(10,19){\circle*{0.5}}
\end{picture}%
}
\setlength{\unitlength}{1pt}
\end{figure}
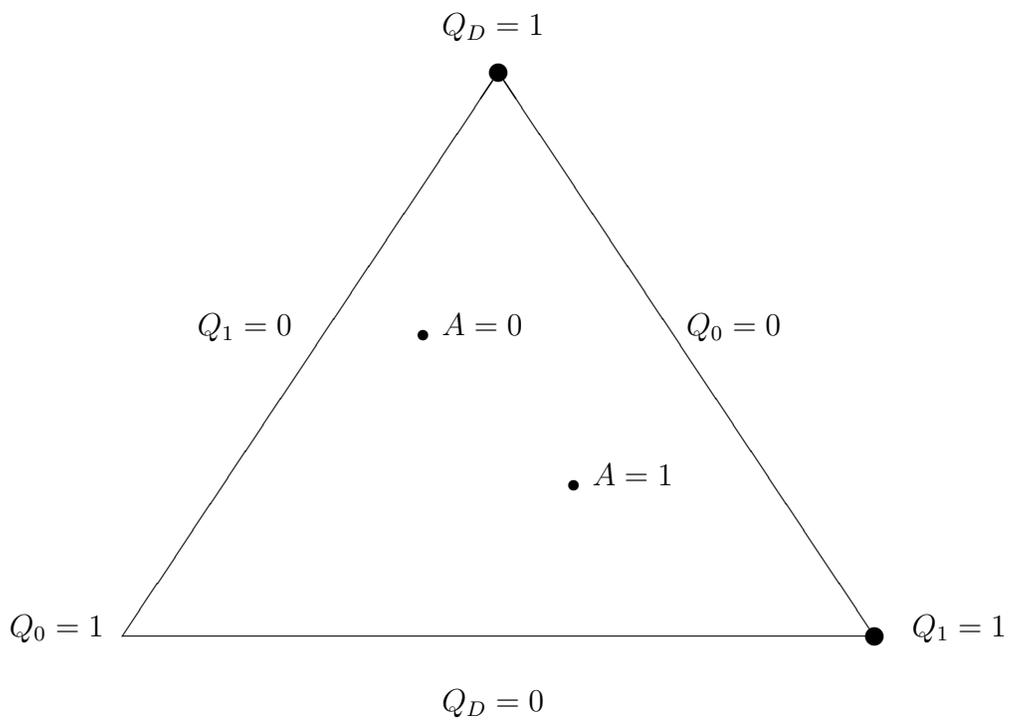%

% \centerline{Figure \ref{fig:simplex} about here}

Note that situations may arise where  treatment may be beneficial in terms of survival while  the opposite can be the case for the  probability $P(Y=1,T>t)$. An example is provided in the Supplementary Materials. This  stresses that one should report both $P(Y^{a}=1,T^{a}>t)$ and $P(T^{a}>t)$ for $a=0,1$. One way to do this is to use the simplex idea as described above. 

\subsection{Real-valued marker}
We now allow $Y$ to be real-valued. Inspired by the previous subsection, and for a fixed time \(t\), define
 $\eta_a(y)=P(T^a>t,Y^a>y)$, $a=0,1$ (suppressing the dependency of $t$ in the notation). We further assume that $L$ is a sufficient set of confounders (no unmeasured confounders) so that the   
 estimands $\{\eta_a(y),P(T^{a}\leq t)\},$ $a=0,1,$ are 
identifiable.
For example,
$$
\eta_a(y)=E\{P(T>t,Y>y|A=a,L)\}.
$$
We use 
$
\eta(y)=\eta_1(y)-\eta_0(y)
$
for the contrast of probabilities.
%Further, let 
%let 
%$\eta_1(y;L)=P(T>t,Y>y|A=1,L)$,
%and define $\eta_0(y;L)$ similarly. 
Note further that $\eta_a(0)={\cal Q}_1^{a}$ if $Y$ is binary. 
 
\section{Debiased learning}\label{sec:eif}
Suppose we have $n$ i.i.d. replicates of $O=\{T_t^*=t\wedge T^*,\Delta_t=I(T\wedge t\leq C),I(t= T_t^*),Y I(t= T_t^*),A,L\}$,
%, where $T$ denotes the survival time, $Y$ is the outcome of interest (QoL) that is measured at landmark time point   $t$ only if the subject has not yet died, $A$  is the treatment indicator (coded 1 for treatment and 0 for control), and $L$ is a vector of baseline covariates. 
assume  also that $0<\pi(L)<1$ a.s., where $\pi(L)=P(A=1|L)$ is the propensity score.
%We use $C$ to denote the censoring variable.
Define $Q_y(A,L)=P(T>t,Y>y|A,L)$, $S(t|A,L)=P(T>t|A,L)$ and $G(y|A,L)=P(Y>y|T>t,A,L)$. %where we have suppressed the dependency on $t$ in the notation. 
We can re-express
$$Q_y(A,L)=G(y|A,L)S(t|A,L).$$

The efficient influence function of $\eta_1(y)$, which forms the basis for the proposed estimation procedure for this quantity, is given in the following theorem.

\begin{thm}
\label{thm:eif}
The efficient influence function of $\eta_1(y)$  w.r.t. the observed data law $P$ of $O=\{T_t^*=t\wedge T^*,\Delta_t=I(T\wedge t\leq C),I(t= T_t^*),Y I(t= T_t^*),A,L\}$, assuming $C\cip \{T,YI(t<T)\}|A,L$, is 
\begin{equation}
    \label{eq:eif}
D^*_{\eta_1}(O,P)  =Q_y(1,L)-\eta_1(y)+\tilde D_{\eta_1}(O,P),
\end{equation}
where $\tilde D_{\eta_1}(O,P)$ is the debiasing term: 
$$
    \tilde D_{\eta_1}(O,P)=\frac{A}{\pi(L)}\left\{\frac{I(t<T^*)I(Y>y)}{K(t|A,L)}+G(y|A,L)S(t|A,L)\int_0^{ t}\frac{dM_C(r|A,L)}{H(r|A,L)}-Q_y(A,L)\right\}
$$
with  $S(r|A,L)=P(T>r|A,L)$, $K(r|A,L)=P(C>r|A,L)$, $H=SK$ 
and where $M_C(r|A,L)$ is the censoring martingale at time $r$ (i.e., the counting process martingale corresponding to the counting process $N_C(r)=I(T^*\leq r,\Delta=0)$ and the filtration generated by this counting process and $A,L$).
\end{thm}

{\bf Remarks}
\vspace{-0.5cm}
\begin{itemize}
\item[(i)] 
If the treatment is randomized with a known randomization probability $P(A=1|L)=\pi$ then, as shown in the Supplementary Material, the efficient influence curve is given by 
\begin{align}
\label{eif:rand}
D^*_{\eta_1}(O,P)-\frac{(A-\pi)}{\pi(1-\pi)}E\left\{D^*_{\eta_1}(O,P)(A-\pi)|L\right\}=D^*_{\eta_1}(O,P),
\end{align}
i.e., it is unchanged. This is because the efficient influence curve $D^*_{\eta_1}(O,P)$ is already orthogonal to functions of the type $(A-\delta)h(L)$ when $A$ and $L$ are independent. The fact that the efficient influence function is unchanged does not mean that the covariates $L$ can be dropped, which is clear from  \eqref{eq:eif}. In other words, had we calculated the efficient influence curve ignoring $L$ (and assuming a correspondingly simpler independent censoring mechanism), then doing the projection corresponding to the left-hand side of \eqref{eif:rand} would give us the already obtained EIF.
\item[(ii)] 
An alternative estimand of potential  interest is $\psi=\psi_1-\psi_0$ with 
$$
\psi_a=E\{Y^{a}I(T^{a}>t\}=E[E\{YI(T>t)|A=a,L\}].
$$
%Since
%\begin{align*}
%    E\{YI(T>t)|A=a,l\}&=\int y P(Y=y|T>t,A=a,l)dyP(T>t|A=a,l)\\
%    &=\int P(Y>y|T>t,A=a,l)dyP(T>t|A=a,l)\\
%    &=\int P(Y>y,T>t|A=a,l)dy
%\end{align*}
%it follows that 
It may be seen that
\begin{equation}
\label{mean}
\psi=\int \{\eta_1(y)-\eta_0(y)\}dy
\end{equation}
and thus also 
$$
D^*_{\psi}(O,P)=\int \{D^*_{\eta_1(y)}(O,P)-D^*_{\eta_0(y)}(O,P)\}dy,
$$
 where we have emphasized  that $D^*_{\eta_a(y)}(O,P)$ depends on $y$, $a=0,1$.

%INTERNAL NOTE: One could also work directly with $E[E\{YI(T>t)|A=a,L\}]$, which leads to the EIF (when no censoring):
%$$
%E\{YI(T>t)|A=a,L\}-\psi_a+\frac{I(A=a)}{P(A=a|L)}\bigl [YI(T>t)-E\{YI(T>t)|A=a,L\}\bigr ]
%$$
%that can be generalized to the censored data case (Tsiatis Ch. 10, as done in the Appendix for the net benefit estimand). However, there should be no need for this, because of \eqref{mean}.
\item[(iii)] 
%In our second data application concerning data from the FLOW clinical outcome trial, see Section \ref{sec:flow}, 
In some applications,
the marker $Y$ may missing for some individuals still alive at the landmark time point. If we let $R$ be the indicator of $Y$ not being missing and assume that $P(R=1|A,L,T^*>t,Y)=P(R=1|A,L,T^*>t)$ (MAR) then the efficient influence function is as in \eqref{eq:eif} except that the debiasing term 
$ \tilde D_{\eta_1}(O,P)$ is changed to 
$$
   \frac{RA}{p(A,L)\pi(L)}\left\{\frac{I(t<T^*)I(Y>y)}{K(t|A,L)}+G(y|A,L)S(t|A,L)\int_0^{ t}\frac{dM_C(r|A,L)}{H(r|A,L)}-Q_y(A,L)\right\},
$$
where $p(A,L)=P(R=1|A,L,T^*>t)$. A proof of this extended result is given in the %Appendix. 
Supplementary Material.
%NOTE: It may be preferable to leave this out altogether; otherwise, we may need to provide remainder terms for this extended situation, which risks sidetracking the main story.
\end{itemize}

\section{Robust estimation and large sample results}\label{sec:estimation}

Since we have a model-free expression of $\eta_1(y)$, it is tempting to substitute 
$S(t|A=1,L)$ and $G(y|A=1,L)$ in $Q_y(1,L)$ by
data-adaptive estimators (e.g., machine learning based estimators)
to obtain an estimator 
$$Q^n_y(1,L)=G_n(y|T>t,A,L)S_n(t|A,L)\}$$ 
of $Q_y(1,L)$, based on which a plug-in estimator
$$
\hat\eta_1^s(y)=\mathbb{P}_n\{Q^n_y(1,L)\},
$$
of $\eta_1(y)$ is readily obtained; here, $\mathbb{P}_n\{v(Z)\}=n^{-1}\sum_i v(Z_i)$ is the empirical measure.
%, and similarly with other unknown quantities. 
Unfortunately, however, the resulting estimator generally has poor performance because the bias-variance trade-offs made in the estimation of $S(t|A=1,L)$ and $G(y|A=1,L)$ are not optimized towards the estimation of $\eta_1(y)$. To first order, the bias of this plug-in estimator is given by minus the sample average of the (estimated) efficient influence curve. 
A one-step (debiased) estimator  \citep{kennedy2022semiparametric} 
 is  therefore  by 
\begin{align*}
  \hat\eta_1^{os}(y)=\hat\eta_1^s(y)
  +\mathbb{P}_n\biggl [&\frac{A}{\pi_n(L)}\biggl \{\frac{I(t<T^*)I(Y>y)}{K_n(t|A,L)}\\
  &-G_n(y|T>t,A,L)S_n(t|A,L)\biggl (1-\int_0^{ t}\frac{dM^n_C(r|A,L)}{H_n(r|A,L)}\biggr )
  \biggr \}\biggr ]. 
\end{align*}
For this, we assume that all the nuisance parameters have been substituted with consistent estimators.
%obtained on a separate part of the data (that excluded observation $i$). 
We therefore make use of data-adaptive estimators $\pi_n(L)$, $G_n(y|A,L)$, $S_n(r|A,L)$ and $K_n(r|A,L)$ of $\pi(L)$, $S(t|A,L)$, $G(y|A,L)$ and $K(t|A,L)$, respectively. 
Specifically, we use super-learners %\citep{laan07:super_learn},
\citep[Chapter3]{targetedlearning_2011}, \citep{westling2024inference} combining parametric and machine learning models.
% R-packages {\tt SuperLearner} and  {\tt survSuperLearner}.
We also make use of 5-fold cross-fitting \citep{chernozhukov_double/debiased_2018}, conducted as explained in \cite{vansteelandt2022assumption}.
Theorem \ref{thm:remainder}, whose proof is provided in the Supplementary Material, states the limit distribution of the resulting estimator.

\begin{thm}\label{thm:remainder}
Consider $n$ i.i.d. replicates of $O=\{T_t^*=t\wedge T^*,\Delta_t=I(T\wedge t\leq C),I(t= T_t^*),Y I(t= T_t^*),A,L\}$
%data $O_i=\{T_{ti}^*=T_i\wedge t\wedge  C_i,\Delta_{ti}=I(T_i\wedge t\leq C_i),I(t= T_{ti}^*),Y_i I(t= T_{ti}^*),A_i,L_i\}$, $i=1,...,n$ 
%$O_i\equiv   (T_i^*,\Delta_i,Y_i I(t\leq T_i^*),A_i,L_i), i=1,...,n$ 
and suppose that $C\cip \{T,YI(t<T)\}|A,L$. Then $\hat\eta_1^{os}(y)$ is asymptotically normally distributed with mean $\eta_1(y)$
and a variance that can be consistently estimated as the sample variance of $D^*_{\eta_1}(O,P_n)$,
provided that the nuisance parameter estimators 
$\pi_n(L)$, $G_n(y|A,L)$, $S_n(r|A,L)$ and $K_n(r|A,L)$ are trained on a separate part of the data, and that   all of the following terms are $o_p(n^{-1/2})$:
\begin{align}
    &E\left [ \left (\frac{\pi_n-\pi}{\pi_n}\right )(L)S_n(t|1,L)(G_n-G)(y|1,L)\right ],
    E\left [ \left (\frac{\pi_n-\pi}{\pi_n}\right )(L)G(y|1,L)(S_n-S)(t|1,L)\right ],\notag\\
    &E\left [ S_n(t|1,L)(G_n-G)(y|1,L)\left (\frac{K_n-K}{K_n}\right )(t|1,L)\right ],\notag\\
    &E\left [ G(y|1,L)(S_n-S)(t|1,L)\left (\frac{K_n-K}{K_n}\right )(t|1,L)\right ], \label{robprop}\\
    &E\left [\pi(L)Q_n(1,L)\int_0^t
    \left (\frac{K}{K_n}\right )(r|1,L)\left(\frac{S_n-S}{S_n}\right )(r|1,L)d(\Lambda_C^n-\Lambda_C)(r|1,L)
    \right ].
    \notag 
\end{align}
\end{thm}

{\bf Remarks}
\vspace{-0.5cm}
\begin{itemize}
\item[(i)] The robustness properties of the estimator  $\hat\eta_1^{os}(y)$ are seen from display \eqref{robprop}
because of the product structure of each of the five terms in \eqref{robprop}. 
They suggest forms of rate double robustness between $\pi_n(L)$ and $G_n(y|1,L)$, between $\pi_n(L)$ and $S_n(t|1,L)$, between $G_n(y|1,L)$ and $K_n(t|1,L)$,
and between $K_n(r|1,L)$ and $S_n(r|1,L)$ $(r\leq t$), meaning, for example, that the rate conditions of Theorem \ref{thm:remainder} can still be achieved when $K_n(t|A,L)$ is slowly converging, provided that this is compensated by fast convergence of $G_n(y|A,L)$.
\item [(ii)] Clearly, a similar result holds for $\hat\eta_0^{os}(y)$  replacing "1" with  "0"  in display \eqref{robprop}. 
 Thus, under these conditions (for both $a=0,1$), also $\hat\eta^{os}(y)=\hat\eta_1^{os}(y)-\hat\eta_0^{os}(y)$ is asymptotically normal distributed with mean $\eta(y)$
and variance that can be consistently estimated as the sample variance of $D^*_{\eta_1}(O,P_n)-D^*_{\eta_0}(O,P_n)$,
\item[(iii)] The efficient influence curve  of $S_a(u)\equiv P(T^a>u)=E\{S(u|a,L)\}$, with $S(u|a,L)=P(T>u|A=a,L)$ $(u\leq t)$, is given by 
$$
D^*_{S_a(u)}(O,P)=S(u|a,L)-S_a(u) -\frac{I(A=a)S(u|a,L)}{P(A=a|L)}\int_0^u\frac{dM(r|a,L)}{H(r|a,L)},
$$
where $dM(r|A,L)$ is the (increment) of the counting process martingale corresponding to the counting process $N(r)=I(T^*\leq r,\Delta=1)$.
The efficient estimator of $S_a(u)$ is then given by 

\begin{equation}
\label{Eff.KM}
\hat S_a(u)=\mathbb{P}_n\left [ S_n(u|a,L)\left \{1-\frac{I(A=a)}{P(A=a|L)}\int_0^{u}\frac{dM_n(r|a, L)}{S_n(r|a,L)K_n(r|a,L)}
\right\}\right ],    
\end{equation} 
where $dM_n(r|a, L)$ is the estimated $dM(r|A,L)$. 
It further follows that $\hat S_a(u)$ is asymptotically normally distributed with mean $S_a(u)$ and variance that can be consistently estimated as the sample variance of $D^*_{S_a(u)}(O,P_n)$. 
% Representaion of EIF using censoring martinagle removed.
%One may rewrite $D^*_{S_a(t)}(O,P)$  as
%$
%D^*_{S_a(t)}(O,P)=S(t|a,L)-S_a(t) +\tilde D_{S_a(t)}(O,P),
%$
%where 
%$$
%\tilde D_{S_a(t)}(O,P)=\frac{I(A=a)}{P(A=a|L)}\left\{\frac{I(T>t)\Delta}{K(T|a,L)}+S(t|a,L)\int_0^{ t}\frac{dM_C(r|a,L)}{H(r|a,L)}-S(t|a,L)\right\}
%$$
%with $\Delta=I(T\leq C)$.

%$$
%S_n(t|a,L)\left \{1-\int_0^{t}\frac{dM_n(r|a, L)}{S_n(r|a,L)K_n(r|a,L)}\right\}
%-\hat S_a(t).$$
%\item [(iv)] 
%Because of \eqref{mean}, 
%$$\hat\psi^{os}=\int \hat\eta^{os}(y) dy
%$$
%that  is asymptotically normal distributed with mean $\psi$ and  variance that can be consistently estimated as the sample variance of 
%$\int \{D^*_{\eta_1(y)}(O,P_n)-D^*_{\eta_0(y)}(O,P_n)\}dy$. (MAYBE assumptions need to be checked here, uniform in $y$?)
\end{itemize}

%\include{examples} 
%KKH and CP decides which example to include.
%% OBS: examples.tex only inlcudes first example by KKH+CP. Old file, 
%% examples.tex, is now called examples_old_aug13.tex

\section{Numerical studies}\label{sec:numerical}

\subsection{Simulation study}\label{sec:sim}

To demonstrate the methodology, we consider a simulation study of three different
scenarios. The first scenario is calibrated to imitate the FLOW study considered in Section~\ref{sec:flow}.

Here we consider a 1:1 randomization, $\pi = P(A=1) = 0.5$,
(A:=1, active, and A:=0, placebo)
with two baseline covariates, $L_2$ a binary indicator for other treatment
usage, $P(L_2=1) = 0.16$,
and a clinical outcome at baseline, $L_1$ (eGFR),
$$
L_1\mid L_2 = 0 \sim \mathcal{N}(46, 225),
\quad
L_1\mid L_2 = 1 \sim \mathcal{N}(51, 235).
$$
The model for the clinical outcome given no truncation at the landmark-time (2
years), is given by
\begin{align*}
&Y\mid L_1, L_2, A = 0 \sim \mathcal{N}(40 + 0.90 L_1 + 2.0 (L_2-\E L_2), 140), \\
&Y\mid L_1, L_2, A = 1 \sim \mathcal{N}(51 + 0.86 L_1 + 2.6 (L_2-\E L_2), 148).
\end{align*}
The event time, $T$ (time to 4-component chronic kidney disease (CKD) or all-cause death), is distributed according to a Weibull model with cumulative hazard given by
\begin{align*}
  \Lambda_T(u | L_1,L_2,A=0) = \exp\{-12.8 -0.023 L_1 -0.56 (L_2-\E L_2)\} u^{1.64} \\
  \Lambda_T(u | L_1,L_2,A=1) = \exp\{-13.0 -0.020 L_1 -0.23 (L_2-\E L_2)\} u^{1.64}
\end{align*}
We consider a scenario for the censoring distribution to illustrate the
properties of the estimator under more extreme right-censoring than was observed
in the original study, which was dominated by administrative censoring. Here we
let the censoring mechanism follow a Weibull model with cumulative hazard
\begin{align*}
  \Lambda_C(u | A) = \exp(-20)u^{2.7 + 0.2A}.
\end{align*}
For the nuisance models, we used a Cox-model for both the censoring and
time-to-event models stratified by treatment and with main effects of $L_1$ and
$L_2$. The outcome ($Y>y$) model was modelled using a Probit regression with
main effects of $L_1$ and $L_2$ and treatment interaction.
We considered the landmark time $t=2$ years and an eGFR cut-point $y=45$.

The simulations were repeated 10,000 times with a sample-size of $n=500$ and $n=1000$, and
results are summarized in Table~\ref{tab:sim1}. All estimates are essentially
unbiased with coverage close to the nominal level. Furthermore, we see large
efficiency gains with the one-step estimator for $\eta_a(y)$ compared to the
unadjusted estimator,
$\{\widehat{K}(t|A=a)\widehat{\pi}_a\}^{-1}\mathbb{P}_n I(A=a,T^\ast>t,Y>y)$.
However, in this setting, where the covariate effects are quite weak and with
low event rates, we do not see such gains for the survival probabilities $S_a$.

In the second scenario, we consider a stronger covariate effect on the
time-to-event process with cumulative hazards given by
\begin{align*}
  \Lambda_T(u | L_1,L_2,A=0) = \exp\{-20.8 -0.83 L_1 -0.56 (L_2-\E L_2)\} u^{1.64} \\
  \Lambda_T(u | L_1,L_2,A=1) = \exp\{-21.0 -0.75 L_1 -0.23 (L_2-\E L_2)\} u^{1.64}.
\end{align*}
The results are summarized in Table~\ref{tab:sim2} for $n=1000$ and conclusions are similar
except that in this scenario, we see a substantial improvement for the one-step
estimator of the survival probabilities with 25\% smaller standard errors
compared to the unadjusted estimates.

In the final scenario 3, we consider a non-randomized situation with
\begin{align*}
P(A|L_1,L_2) = \operatorname{expit}(1 + 0.025 L_1 - 0.5L_2).
\end{align*}
Further, in this scenario, we consider data-adaptive methods for all the
nuisance models. For the outcome model $P(Y>y|T>t,L_1,L_2,A)$ we use a
super-learner with 5-fold cross-validation and three base learners: empirical
mean in each treatment strata, a logistic regression model including main
effects of $L_1, L_2$ and treatment interactions, and a multivariate adaptive
regression spline (MARS) including all second-order interactions \citep{mars:r}.
Similarly, for the treatment model $P(A=1|L_1,L_2)$ we use super-learner with
5-fold cross-validation and three base-learners: empirical mean, logistic
regression with main effects of $L_1, L_2$, and a MARS model with second-order
interactions given $L_1, L_2$. For the time-to-event and censoring models we
use the super-learner algorithm described in \citep{westling2024inference} and
implemented in the \texttt{survSuperLearner} package. We apply 5-fold
cross-validation with the base models:
Kaplan-Meier estimator stratified by treatment, Cox model including main effects of $L_1$ and
$L_2$ and stratified by treatment, and a random survival forest
\citep{ranger-r} including all covariates $A, L_1, L_2$ with default tuning parameters.
The final one-step estimators are constructed using 5-fold
cross-fitting \citep{chernozhukov_double/debiased_2018}. Results are summarized
in Table~\ref{tab:sim3} for $n=1000$, where we also show the results of using the plugin
estimator ignoring the debiasing term. The one-step estimator still performs well in this scenario, whereas the unadjusted estimator, as expected, exhibits large bias due to the unobserved confounding caused by $L_1$ and $L_2$. The plugin estimator, while not showing large bias, suffers from very poor precision, emphasising the need for debiasing and cross-fitting when machine learning is used for predicting the nuisance model components.

\begin{table}
  \centering
  \caption{\label{tab:sim1}Simulation results in scenario 1 for the one-step
    estimates of the parameters $\eta_a(y) = P(T>t, T>y|A=a)$ and
    $S_a = P(T>t|A=a)$ together with the unadjusted estimates ignoring the
    covariates.}
  \begin{tabular}{lrrrrrrr}
    \toprule
    \multicolumn{8}{c}{$n=500$}  \\
\toprule
  & Mean & Bias & SE & SD & SE/SD & Coverage & Rel.eff\\
\midrule
$\eta_0(y)$ & 0.3597 & -0.0004 & 0.0286 & 0.0286 & 0.9988 & 0.9489 & 1.0000\\
$\eta_0(y)$ \emph{unadj.} & 0.3599 & -0.0002 &  & 0.0320 &  &  & 1.1160\\
$\eta_1(y)$ & 0.4225 & -0.0013 & 0.0333 & 0.0338 & 0.9874 & 0.9429 & 1.0000\\
    $\eta_1(y)$ \emph{unadj.} & 0.4223 & -0.0015 &  & 0.0382 &  &  & 1.1308\\
\midrule
$S_0(y)$ & 0.8741 & 0.0002 & 0.0212 & 0.0213 & 0.9988 & 0.9438 & 1.0000\\
$S_0(y)$ \emph{unadj.} & 0.8741 & 0.0002 &  & 0.0214 &  &  & 1.0055\\
$S_1(y)$ & 0.8931 & 0.0003 & 0.0210 & 0.0210 & 1.0028 & 0.9414 & 1.0000\\
$S_1(y)$ \emph{unadj.} & 0.8930 & 0.0003 &  & 0.0210 &  &  & 1.0010\\
    \toprule

    \multicolumn{8}{c}{$n=1000$}  \\
\toprule
  & Mean & Bias & SE & SD & SE/SD & Coverage & Rel.eff\\
\midrule
$\eta_0(y)$ & 0.3598 & -4e-04 & 0.0203 & 0.0204 & 0.9937 & 0.9483 & 1.0000\\
$\eta_0(y)$ \emph{unadj.} & 0.3598 & -3e-04 &  & 0.0228 &  &  & 1.1170\\
$\eta_1(y)$ & 0.4232 & -6e-04 & 0.0237 & 0.0237 & 0.9990 & 0.9478 & 1.0000\\
    $\eta_1(y)$ \emph{unadj.} & 0.4233 & -5e-04 &  & 0.0271 &  &  & 1.1437\\
\midrule
$S_0(y)$ & 0.8735 & -3e-04 & 0.0151 & 0.0151 & 0.9966 & 0.9482 & 1.0000\\
$S_0(y)$ \emph{unadj.} & 0.8735 & -4e-04 &  & 0.0152 &  &  & 1.0045\\
$S_1(y)$ & 0.8931 & 4e-04 & 0.0149 & 0.0148 & 1.0081 & 0.9471 & 1.0000\\
$S_1(y)$ \emph{unadj.} & 0.8931 & 3e-04 &  & 0.0148 &  &  & 1.0005\\
\bottomrule
\end{tabular}
\end{table}

\begin{table}
  \centering
  \caption{\label{tab:sim2}Simulation results for the one-step estimates of the parameters
    $\eta_a(y) = P(T>t, T>y|A=a)$ and
    $S_a = P(T>t|A=a)$ together with the unadjusted estimates ignoring the
    covariates. Scenario 2 with stronger covariate effects on time-to-event endpoint.
  }
  \begin{tabular}{lrrrrrrr}
\toprule
  & Mean & Bias & SE & SD & SE/SD & Coverage & Rel.eff\\
\midrule
$\eta_0(y)$ & 0.3889 & -0.0002 & 0.0202 & 0.0202 & 1.0004 & 0.9499 & 1.0000\\
$\eta_0(y)$ \emph{unadj.} & 0.3887 & -0.0004 &  & 0.0230 &  &  & 1.1378\\
$\eta_1(y)$ & 0.4497 & -0.0017 & 0.0233 & 0.0234 & 0.9917 & 0.9447 & 1.0000\\
    $\eta_1(y)$ \emph{unadj.} & 0.4499 & -0.0015 &  & 0.0267 &  &  & 1.1402\\
\midrule
$S_0(y)$ & 0.7819 & 0.0006 & 0.0138 & 0.0139 & 0.9893 & 0.9481 & 1.0000\\
$S_0(y)$ \emph{unadj.} & 0.7811 & -0.0002 &  & 0.0186 &  &  & 1.3372\\
$S_1(y)$ & 0.8008 & -0.0002 & 0.0135 & 0.0136 & 0.9956 & 0.9491 & 1.0000\\
$S_1(y)$ \emph{unadj.} & 0.8005 & -0.0004 &  & 0.0184 &  &  & 1.3546\\
\bottomrule
\end{tabular}
\end{table}

\begin{table}
  \centering
  \caption{\label{tab:sim3}Simulation results for scenario 3 (non-randomized).
    $\eta_a(y) = P(T>t, T>y|A=a)$ and $S_a = P(T>t|A=a)$ are the one-step estimates
    based on nuisance models estimated with ML techniques, together with the
    unadjusted estimates (ignoring covariates $L_1,L_2$), and the plugin
    estimates (ignoring the debiasing term). }
  \begin{tabular}{lrrrrrrr}
\toprule
  & Mean & Bias & SE & SD & SE/SD & Coverage & Rel.eff\\
\midrule
$\eta_0(y)$ & 0.3896 & 0.0005 & 0.0268 & 0.0268 & 1.0004 & 0.9494 & 1.0000\\
$\eta_0(y)$ \emph{unadj.} & 0.3177 & -0.0714 &  & 0.0287 &  &  & 1.0732\\
$\eta_0(y)$ \emph{plugin} & 0.3917 & 0.0026 &  & 0.0367 &  &  & 1.3711\\
$\eta_1(y)$ & 0.4533 & 0.0019 & 0.0206 & 0.0206 & 0.9986 & 0.9490 & 1.0000\\
$\eta_1(y)$ \emph{unadj.} & 0.4795 & 0.0281 &  & 0.0226 &  &  & 1.0978\\
$\eta_1(y)$ \emph{plugin} & 0.4472 & -0.0042 &  & 0.0304 &  &  & 1.4769\\
\midrule
$S_0(y)$ & 0.7822 & 0.0009 & 0.0148 & 0.0147 & 1.0078 & 0.9486 & 1.0000\\
$S_0(y)$ \emph{unadj.} & 0.7060 & -0.0753 &  & 0.0271 &  &  & 1.8512\\
$S_0(y)$ \emph{plugin} & 0.7820 & 0.0007 &  & 0.0285 &  &  & 1.9443\\
$S_1(y)$ & 0.8017 & 0.0007 & 0.0134 & 0.0134 & 0.9960 & 0.9444 & 1.0000\\
$S_1(y)$ \emph{unadj.} & 0.8299 & 0.0290 &  & 0.0146 &  &  & 1.0828\\
$S_1(y)$ \emph{plugin} & 0.8010 & 0.0001 &  & 0.0271 &  &  & 2.0164\\
\bottomrule
\end{tabular}
\end{table}

\subsection{Application to A Southwest Oncology Trial}\label{sec:oncology}
To illustrate our proposed method, we use data on quality of life one year after treatment for prostate cancer patients. The data were obtained from a Southwest Oncology  Group trial, specifically a randomized phase III trial comparing docetaxel plus estramustine with mitoxantrone plus prednisone in men with metastatic, hormone-independent prostate cancer. The dataset used here was created by \cite{ding2011identifiability} and contains observations on 487 men aged 47 to 88. Of these patients, 258  were randomly assigned to receive docetaxel plus estramustine and 229 were randomly assigned to receive mitoxantrone plus prednisone. We are interested in comparing these two treatments in terms of health-related quality of life (HRQoL) one year after receiving treatment, which is a score between 0 and 100 with 0  and 100 indicating worst HRQoL and the best HRQoL, respectively. In our illustration we also used a binary version using the cutpoint 50, so that $Y=1$ indicates good HRQoL while $Y=0$ indicates poor HRQoL.
In the first treatment group 130 died (out of 258) before 1 year and in the other treatment group 140 (out of 229) died before 1 year, so truncation by death is an issue in this dataset when the interest is to compare the HRQoL. %{\color{blue}OBS: No censoring here!}
We used the baseline HRQoL-score and age ($L$) to increase efficiency.
%Figure \ref{fig:Data_simplex} shows the proposed one-step estimator of
%$({\cal Q}_0,{\cal Q}_D)$ along with 95\% confidence region for the two treatment groups (left panel of figure). The right panel of
\medskip

\begin{figure}
\includegraphics[scale=1]{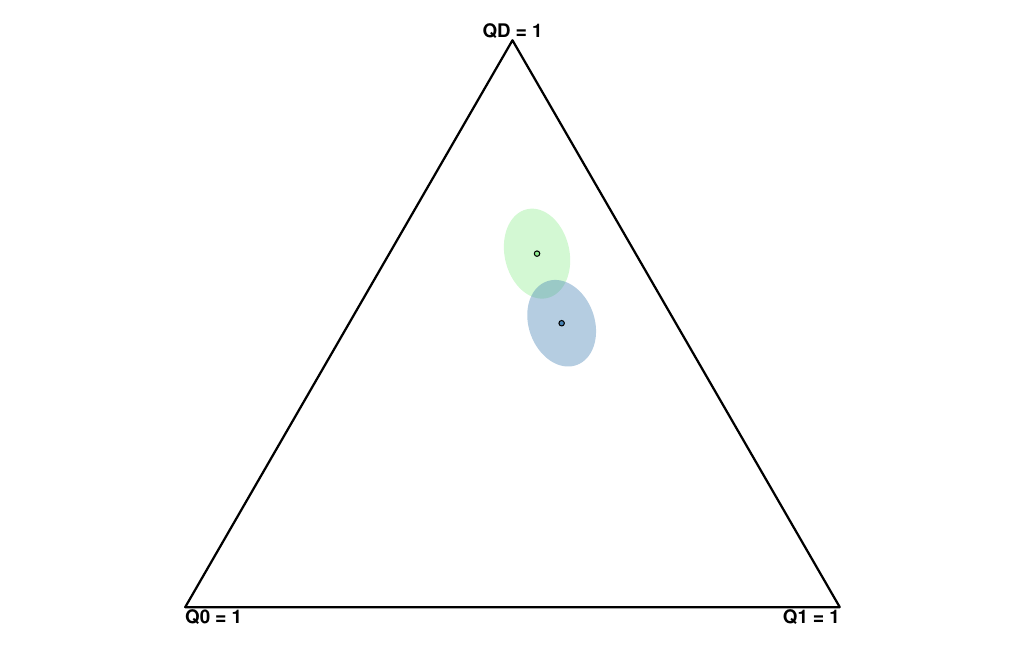}
\caption{Oncology Data. Simplex summarizing $(Q_0,Q_1,Q_D)$ along with 95\% confidence region for the two treatments: docetaxel plus estramustine (steel blue) and mitoxantrone plus prednisone (light green).}
\label{fig:Data_simplex}
\end{figure}%
% \centerline{Figure \ref{fig:Data_simplex}  about here}

Figure \ref{fig:Data_simplex} shows $({\cal Q}_0,{\cal Q}_1,{\cal Q}_D)$ along with a 95\% confidence region for the two treatment groups (green for treatment docetaxel plus estramustine, and blue for the other treatment group).
It is desirable to be away from the top corner  of the simplex (death before 1 year), and to be close to the bottom right corner of the simplex  (alive after 1 year and having a good HRQoL), and we thus see the that treatment with docetaxel plus estramustine, $(\hat {\cal Q}_0,\hat {\cal Q}_1,\hat {\cal Q}_D)=(0.18,0.32,0.50)$,
is superior to treatment with mitoxantrone plus prednisone, $(\hat {\cal Q}_0,\hat {\cal Q}_1,\hat {\cal Q}_D)=(0.15,0.23, 0.62)$.
Even though the 95\% confidence regions overlap slightly, the Wald test for equality of $({\cal Q}_1, {\cal Q}_D)$ for the two treatment arms gives a test-statistic of 8.1 that under the null is approximately $\chi^2$ with two degrees of freedom, resulting in a $p$-value of 0.02.

The chance of being alive after 1 year and having good HRQoL is estimated to be 0.32 when the treatment is  docetaxel plus estramustine, while it is  0.23 for the other treatment. The difference is  0.09 with corresponding  95\% CI: (0.02,0.15), again indicating borderline significance in favor of treatment with docetaxel plus estramustine.
We also calculated the proposed estimator of $\eta(t,y)$ for the different values of
HRQoL-score ($y$-values) in the data, ranging from 0 to 100.
The estimate is shown in Figure \ref{fig:Oncology_eta} along with 95\% confidence limits.

\begin{figure}
\includegraphics[scale=1]{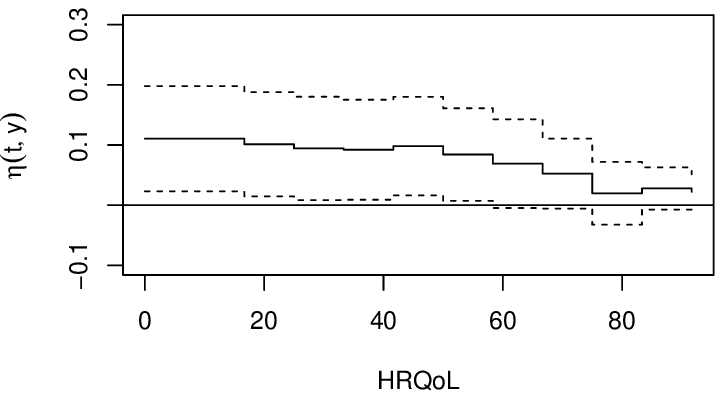}
\caption{Oncology Data. Plot of $\eta(t,y)$ with $t$ equal to 1 year versus $y$ corresponding to HRQoL values as observed in data ranging from 0 to 100 with 0 being poor HRQoL and 100 the highest HRQoL.}
\label{fig:Oncology_eta}
\end{figure}
%\centerline{Figure \ref{fig:Oncology_eta}  about here}
%ALSO(?): $E\{Y^1I(T^1>1)\}-\{Y^0I(T^0>1)\}$

From Figure \ref{fig:Oncology_eta} we see for example that
the probability of being alive after one year and having a HRQoL-score above 20 is estimated to be
0.10  (0.01, 0.19) larger when  treated with group ocetaxel plus estramustine than when treated with mitoxantrone plus prednisone.

These data were also analyzed by \cite{wang2017identification}, who considered the estimand $\delta=E(Y^1-Y^0|T^0\geq t,T^1\geq t)$ with $t$ equal to 1 year and
obtained  $\hat \delta\approx 3$ (-3.97,10.8).
However, this estimand is only identifiable under strong and unverifiable assumptions, and the individuals in the principal stratum $(T^0\geq t,T^1\geq t)$ are unknown.

\subsection{Application to the FLOW clinical kidney outcome
  trial}\label{sec:flow}

The FLOW (Evaluate Renal Function with Semaglutide Once Weekly) clinical trial
aimed to assess kidney outcomes by randomly assigning 3533 participants in a
1:1 ratio to receive either a placebo or a glucagon-like peptide-1 receptor
agonist (GLP-1 RA), alongside standard care \citep{Flow2024}. All enrolled
participants had a diagnosis of type 2 diabetes and were identified as high-risk
for chronic kidney disease (CKD), based on their estimated glomerular filtration
rate (eGFR) calculated from serum creatinine levels and the urinary
albumin-to-creatinine ratio (UACR).

The main outcome measure was the time to the first occurrence of a composite major kidney disease event, which included: a sustained decrease in eGFR of more than 50\% relative to baseline, sustained eGFR drop below 15 mL/min/1.73m², initiation of renal replacement therapy (such as dialysis or transplantation), and renal or cardiovascular death. Additionally, the annual rate of change in eGFR from the point of randomization, termed the total eGFR slope, served as a confirmatory secondary endpoint. The eGFR values were measured at prespecified scheduled visits throughout the trial.

We estimate $P(eGFR(t)>y, T>t)$ and $P(T>t)$ for $y=45$ and $t=2$ years. From a clinical perspective eGFR values below $y=45$ define progression into moderate to severe stage of chronic kidney disease. As for $t=2$ years, this landmark constitutes the standard duration for measuring the impact of treatment on eGFR in CKD trials \citep{stevens2024kdigo}. 

We specify a Cox regression for both the censoring and time-to-event models where we stratify by treatment and include main effects of baseline eGFR and use of SGLT2 inhibitors at baseline. The outcome model $G$ is specified using a logistic regression with treatment specific effects of baseline eGFR and use of SGLT2 inhibitors at baseline. To gauge the efficiency gain from baseline covariate adjustment, we also calculate unadjusted estimates that include only randomized treatment. Owing to randomization, both analysis approaches yield asymptotically unbiased estimates. An overview of analysis results is presented in Table \ref{tab:flowres}. 

\begin{table}
  \centering
  \caption{\label{tab:flowres} Analysis results based on FLOW trial data. Unadjusted estimates are based on randomized treatment alone.
  }
\begin{tabular}{lcccc}
\toprule
\multicolumn{5}{c}{$P(eGFR( t)>y, T> t)$:  Adjusted analysis}\\
Estimand& Estimate & SE & 95\% CI & P-value\\
\midrule
$\eta_{0}(y)$ & 0.3409 & 0.0100 & $[0.3214; 0.3605]$ & -\\
$\eta_{1}(y)$ & 0.3907 & 0.0101 & $[0.3708; 0.4105]$ & -\\
$\eta_{1}(y)-\eta_{0}(y)$ & 0.0497 & 0.0118 & $[0.0265; 0.0729]$ & 0.0000\\

\toprule
\multicolumn{5}{c}{$P(eGFR( t)>y, T> t)$:  Unadjusted analysis}\\
Estimand& Estimate & SE & 95\% CI & P-value\\
\midrule
$\eta_{0}(y)$  & 0.3487 & 0.0112 & $[0.3267; 0.3708]$ & -\\
$\eta_{1}(y)$  & 0.3984 & 0.0117 & $[0.3755; 0.4213]$ & -\\
$\eta_{1}(y)-\eta_{0}(y)$  & 0.0496 & 0.0162 & $[0.0178; 0.0815]$ & 0.0022\\

\toprule
\multicolumn{5}{c}{$P(T> t)$:  Adjusted analysis}\\
Estimand& Estimate & SE & 95\% CI & P-value\\
\midrule
$S_{0}( t)$ & 0.8699 & 0.0080 & $[0.8542;0.8855]$ & -\\
$S_{1}( t)$ & 0.9013 & 0.0071 & $[0.8874; 0.9153]$ & -\\
$S_{1}( t)-S_{0}( t)$ & 0.0315 & 0.0107 & $[0.0106; 0.0524]$ & 0.0032\\
\toprule
\multicolumn{5}{c}{$P(T> t)$:  Unadjusted analysis}\\
Estimand& Estimate & SE & 95\% CI & P-value\\
\midrule
$S_{0}( t)$ & 0.8704 & 0.0080 & $[0.8547; 0.8861]$ & -\\
$S_{1}( t)$  & 0.9008 & 0.0071 & $[0.8868; 0.9148]$ & -\\
$S_{1}( t)-S_{0}( t)$  & 0.0304 & 0.0107 & $[0.0094; 0.0515]$ & 0.0046\\
\bottomrule
\end{tabular}

\end{table}

From Table \ref{tab:flowres} we conclude a clear benefit of active treatment for both survival and kidney function compared to placebo. This conclusion is further supported by the joint graphical evaluation provided in the simplex summary in Figure \ref{Fig_Kidney}. We also note that the proposed covariate adjustment substantially reduces the standard error of the estimated joint probability compared to no adjustment. For survival probability, on the other hand, no noticeable gain in efficiency due to covariate adjustment is observed

\medskip

\begin{figure}
\label{Fig_Kidney}
  \centering
  \includegraphics[width=0.7\textwidth]{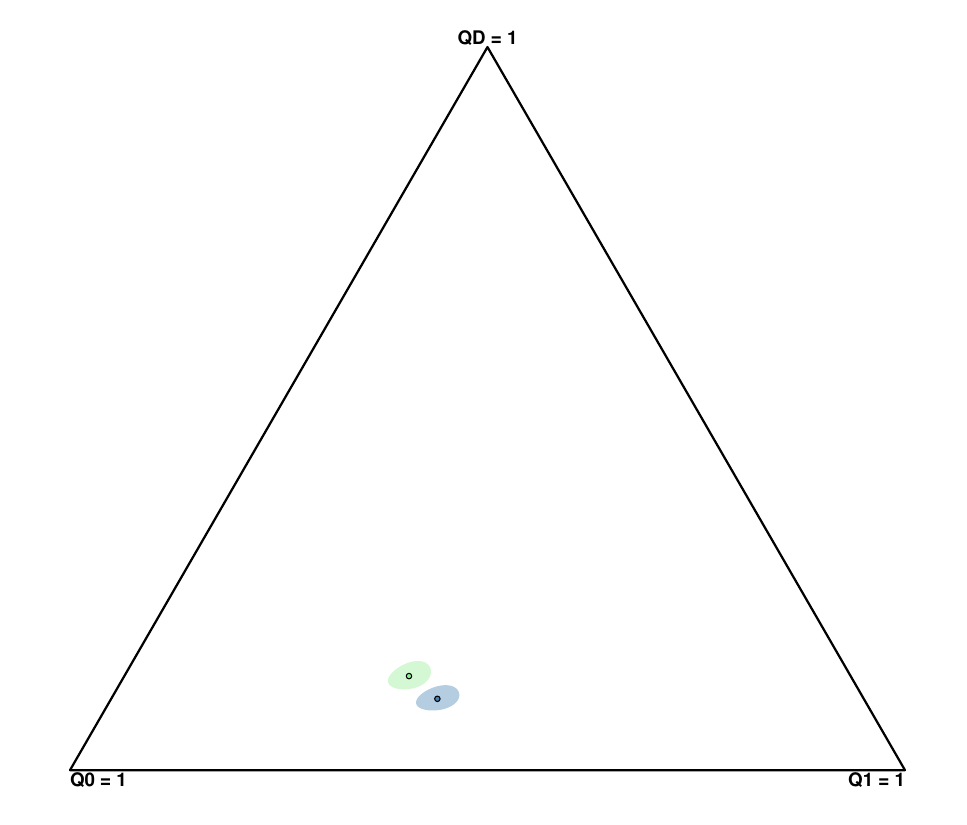}
  \caption{Kidney Data. Simplex summarizing $(Q_0,Q_1,Q_D)$ along with 95\% confidence region for the two treatments: GLP-1 RA + standard of care (steel blue) and placebo + standard of care (light green)}
\end{figure}%
% \centerline{Fig \ref{Fig_Kidney} about here}

\section{Concluding remarks}
In Section 2.2 we suggest to illustrate estimates of the survival probabilities $\mathcal{Q}^{a}_{\mathcal{D}}$ and the joint probabilities $\eta_{a}(y)$ in a simplex to gauge the effect of treatment.
%However, for a formalized assessment of the treatment benefit we would need to
Futher development might focus on
assessing the  one-sided superiority null-hypotheses
\begin{align*}
&H_{1}:\:\:  \mathcal{Q}^{0}_{\mathcal{D}}-\mathcal{Q}^{1}_{\mathcal{D}}\leq 0\\
& H_{2}:\:\:  \eta_{1}(y)-\eta_{0}(y)\leq 0.\\
\end{align*}
In cases where superiority can only be confirmed for one of the hypotheses an additional assessment of overall treatment benefit is warranted. Such an assessment could be specified by a utility in which the sum of weighted treatment contrasts is used to quantify the overall treatment benefit. For some prespecified weight $w\in (0,1)$, this amounts to assesssing the additional superiority null hypothesis:
$$
H_{u}:\:\: w\cdot(\mathcal{Q}^{0}_{\mathcal{D}}-\mathcal{Q}^{1}_{\mathcal{D}})+(1-w)\cdot\{\eta_{1}(y)-\eta_{0}(y)\}\leq 0.
$$
Some explicit strategies on how to assess these three hypotheses under type 1 error control are provided in \cite{pipper2025generalapproachconstructpowerful}.

We have shown how the points $(\mathcal{Q}^{0}_{\mathcal{D}},\eta_0(y))$ and $(\mathcal{Q}^{1}_{\mathcal{D}},\eta_1(y))$ may be depicted in a {\em simplex} for a graphical comparison between treatment groups, and in Figures \ref{fig:Data_simplex} and \ref{Fig_Kidney} the points were equipped with confidence regions for inference. Such a plot may be extended in several directions. Thus, curves are obtained by varying $y$, i.e., plotting $(\mathcal{Q}^{a}_{\mathcal{D}},\eta_a(y_j)), j=1,\dots,k, a=0,1$ for a number of thresholds $y_1,\dots,y_k$. Alternatively, the fixed time point, $t$, may be replaced by a number of time points. Inference for the resulting curves is an open question that may be subject to future research.

\nothere{
Points to consider:

\begin{itemize}
    \item More than one time point: $t_1, \ldots, t_k$
    \item Assessment of the treatment effect based on the joint probability
$P(T>\tau, Y>\mu \mid A=a)$ may be seen as a special-case of a utility
\begin{align*}
c_1 1(T<\tau) + c_2 1(Y>\mu, T>\tau)
\end{align*}
where we do not assign a cost $c_1$ but in a way let the events $(Y>\mu)$ and
$(T>\tau)$ count equally.
\item \cite{bannick2025:adj_rct}
\end{itemize}
}

\bibliographystyle{biometrika}
\bibliography{references}

%\newpage
\clearpage

\pagenumbering{arabic}  % Resets page number to 1 and enables numbering
\pagestyle{headings}
%$\mbox{}$

%\vspace{0.1cm}
\section*{Online Supplementary Material}

%\section*{Appendix}

\noindent{\it Example challenging the interpretation of the joint probability as a stand alone. \quad}
\vspace{-0.3cm}
%\section{An example challenging the interpretation of the joint probability as a stand alone}

%In this section 
This example  leads to a situation where there is a beneficial treatment effect when assessed by the risk of death but a harmful treatment effect when contrasting the  joint probability. 
This  pin points the need for complementing 
%illustrates that for a balanced  assessment of treatment benefit, a claim based on 
the joint probability %should be complemented by 
with the impact of treatment on the risk of death.       The example is built around a binary predictive biomarker $Z$ that is activated under treatment intervention $A=1$ to yield a differentiation in the rate of death. Under treatment intervention $A=0$ the rate of death is not affected by the biomarker. Specifically we assume that 
$
P(Z=z_{1})=P(Z=z_{2})=\tfrac{1}{2}$ for some  $0<z_1<z_2$, and further 
%We further assume 
that treatment intervention $A$ is independent of the biomarker $Z$. Conditional on intervention and biomarker, the time of death $T$ follows a distribution with hazard rates $\alpha(t| A=0, Z)=1$ and $
\alpha(t| A=1, Z)=Z$.
Note that when $z_{1}<1$, a beneficial effect of treatment in terms of reduced risk of death  manifests at landmark time $t>\frac{\log(2)}{1-z_{1}}$, since in that case
$$
\frac{P(T>t| A=1)}{P(T>t|A=0)}=\frac{1}{2}\exp\{(1-z_{1})t\}+\frac{1}{2}\exp\{(1-z_{2})t\}>1.
$$
Furthermore, a selection occurs leading to a lower proportion of subjects with the high biomarker value among survivors in the treatment arm $A=1$ at the landmark time $t$, that is:
\begin{align*}
P(Z=z_{2}| T>t, A=1)&=\frac{P(T>t| Z=z_{2}, A=1)}{P(T>t| Z=z_{1}, A=1)+P(T>t| Z=z_{2}, A=1)}\\
&=
\frac{1}{1+\exp\{ (z_{2}-z_{1})t\}}<\frac{1}{2}.
\end{align*}
For the treatment arm $A=0$ no selection occurs. This entails that treatment will impact a clinical score recorded among subjects without terminal event towards lower values through selection if the clinical score is positively related to the biomarker.        
We construct such a situation by setting the clinical score at time $t$ among survivors to the value of the biomarker, that is, 
\begin{align*}
Y = I(T>t)Z.
\end{align*}
From a straightforward calculation, it is seen that  
\begin{align*}
P(Y\geq z_{2},T>t|A=1)=&P(Y\geq z_{2},T>t|A=0)\times\frac{\exp\{(1-z_{2})t\}+\exp\{(1-z_{1})t\}}{1+\exp\{ (z_{2}-z_{1})t\}}\\
<& P(Y\geq z_{2},T>t|A=0),\:\:z_{2}>1.
\end{align*}
We conclude that when $t>\frac{\log(2)}{1-z_{1}}$ and $z_{2}>1$, there is a beneficial treatment effect judged by the risk of death but a harmful treatment effect judged by the joint probability. 
 \hfill $\square$

\noindent{\it Proof of Theorem \ref{thm:eif}}\quad
\vspace{-0.3cm}

We use $Z$ to denote the "full data" case (no censoring), so $Z=\{YI(T>t),I(T>t);A,L\}$.
It follows immediately using the tips and tricks of \cite{kennedy2022semiparametric} that the efficient influence curve of $\eta_1(y)$  w.r.t.  to the "full data" case  is
$$
D^*_{\eta_1}(Z,P)  =Q_y(1,L)-\eta_1(y)+\tilde D_{\eta_1}(Z,P),
$$
where $\tilde D_{\eta_1}(Z,P)$ is the debiasing term:
$$
    \tilde D_{\eta_1}(Z,P)=\frac{A}{\pi(L)}\left\{
    I(Y>y,T>t)-Q_y(A,L)\right\}
$$
We then apply \cite{tsiatis2006semiparametric} Ch. 10 (see Sect. 10.4) and apply some algebra to get the desired result.
\hfill $\square$

\noindent{\it Proof of Remark (iii) in Section \ref{sec:eif}} \quad
\vspace{-0.3cm}

If we let $Z=(I(T^*>t)Y,A,L)$ denote full data, which is what we refer to as the observed data in Theorem \ref{thm:eif} then the "full data" EIF is $D_{\eta_1}^*(O,P)$ given in that Theorem. The situation now is that $Y$ may be missing even if $T^*>t$ and we use $R$ to denote the the indicator of $Y$ being measured in this situation so that MAR is assumed: $P(R=1|T^*>t,A,L,Y)=P(R=1|T^*>t,A,L)\equiv p(A,L)$. The observed data (meaning that $Y$ may now be missing for some) EIF is then given by
$$
\frac{RD_{\eta_1}^*(O,P)}{p(A,L)}-\Pi\left (\frac{RD_{\eta_1}^*(O,P)}{p(A,L)} |\Lambda_2\right )
$$
where $\Lambda_2$ is so-called augmentation space,
$$
\Lambda_2=\left \{\frac{R-p(A,L)}{p(A,L)}h(A,L):\, \mbox{any }h(A,L)\right \}.
$$
Calculating the projection gives
$$
\Pi\left (\frac{RD_{\eta_1}^*(O,P)}{p(A,L)} |\Lambda_2\right )=
\frac{R-p(A,L)}{p(A,L)}E\{D_{\eta_1}^*(O,P)|A,L)\}.
$$
Now collecting terms  gives the claimed EIF.
 \hfill $\square$

\noindent{\it Proof of Theorem \ref{thm:remainder}}\quad
\vspace{-0.3cm}

We use the notation %$P(A=1|L)=\pi(L),$
$
\psi(P)=E\{Q(1,L)\}.
$
%with
%$$
%Q(1,L)=P(T>t,Y>y|A=1,L).
%$$
In the case where there is no censoring we get the remainder term
\begin{align}
\label{R-nocens}
    R_n=&\psi(P_n)-\psi(P)+E\{D^*(\psi,P_n)\}\nonumber\\
    =&E\left [ \left (\frac{\pi_n-\pi}{\pi_n}\right )(L)\{Q_n(1,L)-Q(1,L)\}\right ]\nonumber\\
    =&E\left [ \left (\frac{\pi_n-\pi}{\pi_n}\right )(L)S_n(t|1,L)(G_n-G)(y|1,L)\}\right ]\nonumber\\
    &+E\left [ \left (\frac{\pi_n-\pi}{\pi_n}\right )(L)
    G(y|1,L)(S_n-S)(t|1,L)\}\right ],
\end{align}
which follows as in the classical ATE example.
Above,
 we have decomposed $Q$ into $GS$,
$$
Q(1,L)=P(T>t,Y>y|A=1,L)=P(Y>y|T>t,A=1,L)P(T>t|A=1,L)\equiv (GS)(1,L),
$$
and used
$$
G_nS_n-GS=(G_n-G)S_n+(S_n-S)G.
$$
From \eqref{R-nocens} we see the robustness properties of the corresponding one-step estimator, and also that asymptotic linearity is obtained if we use machine learning (ML) on $\pi_n,G_n$ and $S_n$, where there is sufficiently fast convergence, ie $n^{1/2}$ times the two terms given in \eqref{R-nocens} needs to go to zero.

We now turn to the case where there is censoring. The EIF is
$$
D^*(\psi,P)=Q(1,L)-\psi(P)+\tilde D(\psi,P),
$$
where the debiasing term $\tilde D(\psi,P)$ is:
$$
    \tilde D(\psi,P)=\frac{A}{\pi(L)}\left\{
    \frac{I(t<T\wedge C)I(Y>y)}{K(t|A,L)}+Q(1,L)\int_0^{ t}\frac{dM_C(r|1,L)}{H(t|1,L)}
    -Q(1,L)\right\},
$$
where  $H=SK$.
The remainder term is
\begin{align*}
    R_n=&\psi(P_n)-\psi(P)+E\{D^*(\psi,P_n)\}\\
    =&E\left [ \left (\frac{\pi_n-\pi}{\pi_n}\right )(L)\{Q_n(1,L)-Q(1,L)\}\right ]
    -E\left [ \left (\frac{\pi}{\pi_n}\right )(L)\left (\frac{K_n-K}{K_n}\right )(t|1,L)Q(1,L)\right ]\\
    &+E\left [ \frac{A}{\pi_n(L)}Q_n(1,L)\int_0^t\frac{dM_C^n(r|1,L)}{H_n(r|1,L)}\right ],
\end{align*}
where the last term on the RHS in the latter display equals
$$
-E\left [ \left (\frac{\pi}{\pi_n}\right )(L)Q_n(1,L)
\int_0^t\left (\frac{H}{H_n}\right )(r|1,L)d(\Lambda_C^n-\Lambda_C)(r|1,L)\right ].
$$
After some further algebra, exploiting that
$$
\frac{H}{H_n}=\frac{K}{K_n}-\frac{K}{K_n}\frac{(S_n-S)}{S_n}
$$
and
$$
\int_0^t\left(\frac{K}{K_n}\right)(r|1,L)d(\Lambda_C^n-\Lambda_C)(r|1,L)=
\left(\frac{K}{K_n}\right)(t|1,L)-1,
$$
it follows that
\begin{align*}
    R_n
    =&E\left [ \left (\frac{\pi_n-\pi}{\pi_n}\right )(L)(Q_n-Q)(1,L)\}\right ]
    +E\left [ (Q_n-Q)(1,L)\left (\frac{K_n-K}{K_n}\right )(t|1,L)\right ]\\
    &-E\left [\pi(L)Q_n(1,L)\int_0^t
    \left (\frac{K}{K_n}\right )(r|1,L)\left(\frac{S_n-S}{S_n}\right )(r|1,L)d(\Lambda_C^n-\Lambda_C)(r|1,L)
    \right ]
    \\
    &+E\left [ \left (\frac{\pi_n-\pi}{\pi_n}\right )(L)(Q_n-Q)(1,L)\left (\frac{K_n-K}{K_n}\right )(t|1,L)\right ].
\end{align*}
Again, decomposing $Q$ into $GS$,
\begin{align*}
Q(1,L)&=P(T>t,Y>y|A=1,L)\\
&=P(Y>y|T>t,A=1,L)P(T>t|A=1,L)\\
&\equiv (GS)(1,L),
\end{align*}
and using
$$
G_nS_n-GS=(G_n-G)S_n+(S_n-S)G,
$$
we can rewrite $R_n$ as
\begin{align*}
    &E\left [ \left (\frac{\pi_n-\pi}{\pi_n}\right )(L)S_n(t|1,L)(G_n-G)(y|1,L)\right ]+E\left [ \left (\frac{\pi_n-\pi}{\pi_n}\right )(L)G(y|1,L)(S_n-S)(t|1,L)\right ]\\
   + &E\left [ S_n(t|1,L)(G_n-G)(y|1,L)\left (\frac{K_n-K}{K_n}\right )(t|1,L)\right ]\\
   +
    &E\left [ G(y|1,L)(S_n-S)(t|1,L)\left (\frac{K_n-K}{K_n}\right )(t|1,L)\right ]\\
   - &E\left [\pi(L)Q_n(1,L)\int_0^t
    \left (\frac{K}{K_n}\right )(r|1,L)\left(\frac{S_n-S}{S_n}\right )(r|1,L)d(\Lambda_C^n-\Lambda_C)(r|1,L)
    \right ]
    \\
    +&E\left [ \left (\frac{\pi_n-\pi}{\pi_n}\right )(L)\{Q_n(1,L)-Q(1,L)\}\left (\frac{K_n-K}{K_n}\right )(t|1,L)\right ].
\end{align*}
Again, from the latter display, we can read off the robustness properties of the corresponding one-step estimator.
Furthermore, asymptotic linearity is obtained if we use ML  on $\pi_n,G_n, K_n$ and $S_n$ assuming sufficiently fast convergence of the ML-methods, i.e., $n^{1/2}$ times each of the  terms given in the latter display  need to go to zero.
If we were to use parametric/semiparametric working models (which we do not recommend), we see the "usual" structure when dealing with right-censored data that if the censoring model is misspecified then we need the (two) outcome models to be correctly specified. It does not suffice to have the propensity score model correctly specified to preserve consistency. Specifically, consistency is attained if either (i) the outcome models ($G_n$ and $S_n$) are correctly specified or if (ii) the propensity score model ($\pi_n$) and the censoring model ($K_n$) are correctly specified. \hfill $\square$

\end{document}